\newif\iflncs
\lncstrue
\lncsfalse

\iflncs
\documentclass[fleqn,11pt,oribibl]{llncs}
\else
\documentclass[letterpaper, fleqn,11pt]{article}
\fi

\usepackage[english]{babel}
\usepackage[boxed]{algorithm}
\usepackage[noend]{algorithmic}
\iflncs
\usepackage{amsmath,amssymb,amsfonts}
\else
\usepackage{amsmath,amssymb,amsfonts,amsthm}
\fi
\usepackage[fleqn,tbtags]{mathtools}
\usepackage{lastpage}
\usepackage{bbm}
\usepackage{fancybox}
\usepackage{fullpage}
\usepackage{boxedminipage}
\usepackage{fancyhdr}
\usepackage{graphicx}
\usepackage{subfigure}
\usepackage{epsfig}
\usepackage{color}
\usepackage{xspace}
\usepackage{authblk}
\usepackage{url}
\usepackage [autostyle, english = american]{csquotes}
\usepackage{ifpdf}

\ifpdf    
\usepackage{hyperref}
\else    
\usepackage[hypertex]{hyperref}
\fi

\MakeOuterQuote{"}

\iflncs
\else
\newtheorem{theorem}{Theorem}[section]
\newtheorem{lemma}[theorem]{Lemma}
\newtheorem{corollary}[theorem]{Corollary}
\newtheorem{proposition}[theorem]{Proposition}
\newtheorem{claim}[theorem]{Claim}

\fi

\newtheorem{notation}{Notation}

\makeatletter
\newtheorem*{rep@theorem}{\rep@title}
\newcommand{\newreptheorem}[2]{%
\newenvironment{rep#1}[1]{%
 \def\rep@title{#2 \ref{##1}}%
 \begin{rep@theorem}}%
 {\end{rep@theorem}}}
\makeatother

\newreptheorem{theorem}{Theorem}

\newcommand{\F}{\mathbb{F}}
\newcommand {\ignore} [1] {}

\def \eqdef {:=}

\def \ind {i^*}

\renewcommand{\dagger}{\intercal}

\DeclareMathOperator{\modu}{mod}
\renewcommand{\dim}{{\sf dim}}

\providecommand{\eqdef}{:=}

\newcommand{\etal}{{\em et al.\ }\xspace}

\iflncs
\else
\fi

\title{Tight Cell Probe Bounds for Succinct Boolean Matrix-Vector Multiplication}

\author{Diptarka Chakraborty\thanks{Computer Science Institute of Charles University, Prague. \texttt{diptarka@iuuk.mff.cuni.cz}. Supported by the funding from the European Research Council under the European Union's Seventh Framework Programme (FP/2007-2013)/ERC Grant Agreement no. 616787.} \qquad Lior Kamma\thanks{MADALGO. Aarhus
    University. \texttt{lior.kamma@cs.au.dk}. Supported by a Villum Young Investigator Grant.} \qquad Kasper Green
  Larsen\thanks{MADALGO. Aarhus
    University. \texttt{larsen@cs.au.dk}. Supported by a Villum Young
    Investigator Grant and an AUFF Starting Grant.}}

\begin{document}

\date{}
\maketitle

\begin{abstract}
The conjectured hardness of Boolean matrix-vector multiplication has
been used with great success to prove conditional lower bounds for
numerous important data structure problems, see Henzinger et
al. [STOC'15]. In recent work, Larsen and Williams [SODA'17] attacked
the problem from the upper bound side and gave a surprising cell probe
data structure (that is, we only charge for memory accesses, while computation is free). Their cell probe data structure answers queries in
$\tilde{O}(n^{7/4})$ time and is succinct in the sense that it stores the
input matrix in read-only memory, plus an additional
$\tilde{O}(n^{7/4})$ bits on the side. In this paper, we essentially
settle the cell probe complexity of succinct Boolean matrix-vector
multiplication. We present a new cell probe data structure with query
time $\tilde{O}(n^{3/2})$ storing just $\tilde{O}(n^{3/2})$ bits on
the side. We then complement our data structure with a lower bound
showing that any data structure storing $r$ bits on the side, with $n
< r < n^2$ must have query time $t$ satisfying $t r =
\tilde{\Omega}(n^3)$. For $r \leq n$, any data structure must have $t
= \tilde{\Omega}(n^2)$. Since lower bounds in the cell probe model
also apply to classic word-RAM data structures, the lower bounds naturally
carry over. We also prove similar lower bounds for
matrix-vector multiplication over $\mathbb{F}_2$.
\end{abstract}
\newpage

\section{Introduction}
\label{sec:intro}
Matrix-vector multiplication is one of the most fundamental
algorithmic primitives. In the data structure variant of the problem,
we are given an $n \times n$ matrix $M$ as input. The goal is to preprocess
$M$ into a data structure, such that upon receiving any $n$-dimensional
query vector $v$, we can quickly compute $Mv$. Constructing a data
structure instead of computing $Mv$ directly may pay off as soon as
we have to answer multiple matrix-vector multiplication queries on the
same matrix $M$.

When defined over the Boolean semiring (with addition replaced by OR
and multiplication replaced by AND) the above problem is a special
case of the well-known \emph{Online Matrix-Vector (OMV)} problem:
Given a matrix $M \in \{0,1\}^{n \times n}$ and a stream of vectors
$v^{(1)},\cdots, v^{(n)} \in \{0,1\}^n$ the goal is to output the
value of $Mv^{(i)}$ before seeing $v^{(j)}$ for any $j \in
\{i+1,\cdots ,n\}$. Henzinger \emph{et al.}~\cite{HKNS15} conjectured
that the OMV problem cannot be solved by any randomized algorithm with
error probability at most $1/3$ within time $O(n^{3 -\epsilon})$ for
any constant
$\epsilon > 0$ (i.e. amortized $O(n^{2-\epsilon})$ per vector
$v^{(i)}$). This conjecture is known as the OMV conjecture and is one 
of the central conjectures in the ``Hardness
in P'' area, along with the Strong Exponential Time Hypothesis
(SETH see~\cite{ImpagliazzoPaturi}), the 3SUM conjecture (see e.g.~\cite{AnkaOvermars}) 
and the All Pairs Shortest Paths conjecture (APSP see e.g.~\cite{WAPSP}). 
The conjecture implies a whole range of near-tight conditional
lower bounds for classic (fully/partially) dynamic data structure
problems such as dynamic reachability. For the matrix-vector
multiplication problem over the Boolean semiring, the OMV conjecture in
particular implies that for any polynomial preprocessing time and
space, the query time must be $n^{2-o(1)}$~\cite{HKNS15}.

The current best upper bound for the OMV problem is due to Larsen and
Williams~\cite{LW17}, who gave a randomized (word-RAM) 
data structure with a total running time
of $n^3/2^{\Omega(\sqrt{\lg n})}$ over a sequence of $n$
queries, i.e. amortized $n^2/2^{\Omega(\sqrt{\lg n})}$ per
query. While this new upper bound is non-trivial, it does not violate
the OMV conjecture.

The holy grail is, of course, to replace the OMV conjecture by an
unconditional $n^{3-o(1)}$ lower bound. In contrast to SETH, 3SUM
and APSP, the OMV problem is a data structure problem, rather than an
algorithmic problem. Since we have been vastly more successful in
proving unconditional lower bounds for data structures than for
algorithms, it does not a priori seem completely hopeless to prove a tight
unconditional lower bound for OMV in the foreseeable future. Data
structure lower bounds are typically proved in the cell probe model of
Yao~\cite{Yao81a}. In this model, computation is free of cost, and the
complexity of a data structure is solely the amount of memory it uses and 
the number of memory accesses it performs on answering a query. In
particular, lower bounds proved in the cell probe model apply to
data structures developed in the standard word-RAM model,
regardless of which unit cost instructions are available. Quite
surprisingly, Larsen and Williams~\cite{LW17} showed that the performance
of their OMV data structure greatly improves if implemented in the
cell probe model (i.e. if computation is free and we only charge for
accessing memory). Their cell probe data structure for matrix-vector
multiplication over the Boolean semiring has a query
time of $O(n^{7/4}/\sqrt{w})$, where $w$ denotes the
word size (typically $w = \Theta(\lg n)$). Thus the OMV conjecture is
\emph{false} in the cell probe model! 

But what is then the true complexity of matrix-vector multiplication in the cell probe model?
Is it the best one can hope for, namely $O(n/w)$ time (the size of the
output measured in words) with $O(n^2)$ bits of space? Or does one
have to pay a polynomial factor in either space or time? While not
matching the OMV conjecture, a polynomial lower bound ($n^{1+\Omega(1)}$ query time with, 
say, polynomial space) would still be immensely valuable as it would
be the first polynomial lower bound for any data structure problem and
would be a huge leap forward in proving unconditional lower
bounds. Furthermore, it would still imply non-trivial polynomial lower
bounds for numerous important data structure problems via the
reductions already given in previous papers. 

Our main contribution is to prove (near-)tight polynomial cell probe lower
bounds for the class of \emph{succinct} matrix-vector multiplication data structures. 
Before formally presenting our
results, we survey the current barriers for proving cell probe lower
bounds as this will help understand the context of our results.

\paragraph{Cell Probe Lower Bound Barriers.}
Much effort has gone into developing techniques for proving data
structure lower bounds in the cell probe model. For static data
structures (like matrix-vector multiplication), the current strongest
techniques~\cite{Larsen:2012:focs} can prove lower bounds of $t = \Omega(\lg
m/\lg \alpha)$ where $t$ is the query time, $m$ is the number of distinct possible queries in the
problem and $\alpha$ is the space-overhead over linear. Thus the
strongest previous lower bounds peak at $t = \Omega(\lg m)$ with
linear space. For matrix-vector multiplication with an $n \times n$
matrix $M$ (over $\{0,1\}$), there are $2^n$ possible queries $v$, thus the strongest
possible lower bound current techniques would allow us to prove is
$t = \Omega(n)$. This is unfortunately not much more than the
trivial $t = \Omega(n/w)$ one would get for just writing the output.

For dynamic data structures, i.e. data structures where one receives
both updates to the input data and queries, the current strongest
techniques~\cite{Larsen12a, weinstein:lbs} give lower bounds of $t = \Omega(\lg m \lg n /
\lg^2(u w))$ where $u$ is the update time, $w$ the word-size and
$n$ the input size/number of updates performed. This is about a $\lg
n$ factor more than for static data structures, thus still leaves us
quite far from proving lower bounds close to the conditional
ones. If we restrict ourselves to lower bounds for decision problems,
i.e. problems where the answer to a query is just one bit, the
situation is worse, with the strongest lower bounds being of the form
$t = \Omega(\lg m \sqrt{\lg n}/\lg^2(uw))$~\cite{LWY17}.

\paragraph{Restricted Data Structures.}
For many data structure problems, the lower bounds one can prove with
previous techniques are in fact tight, see e.g.~\cite{FS89, PD06,
  Pat11, PT11,Larsen12a, weinstein:lbs}. However, as we can see from
matrix-vector multiplication, there are also problems where the
current techniques are quite far from proving what we believe should
be the right lower bound. This has resulted in researchers proving a number of exciting
lower bounds for special classes of data
structures. For instance, Clifford \etal ~\cite{CGL15} consider matrix-vector
multiplication over a finite field $\F_p$ of exponential size $p =
2^{\Theta(n)}$. This results in more queries to the problem ($p^n =
2^{\Theta(n^2)}$) and thereby enabled proving a lower bound of $t =
\Omega(n^2/\lg \alpha)$. Their lower bounds hold when the word size $w$ is big enough to
store an element of the field, i.e. $w = \Theta(n)$ bits. An
interesting interpretation of the lower bound is that, as long as we do not take advantage of the
size of the field, any data structure is bound to use $\Omega(n^2/\lg
\alpha)$ query time. Another line of work has focused on non-adaptive
dynamic data structures~\cite{brody, brodynew, Ramamoorthy2017NonAdaptiveDS}. These are data structures where the
memory locations read upon answering a query depend \emph{only} on
the query and \emph{not} on the contents of probed cells, i.e. the
query algorithm may not branch based on what it reads. 

\paragraph{Succinct Data Structures.}
The last class
of data structures we consider are \emph{succinct} data
structures. Succinct data structures use
space very close to the information theoretic minimum. More formally,
we say that a data structure has {\em redundancy} of $r$ bits if its
space usage is $\Pi + r$ bits where $\Pi$ is the information theoretic
minimum for solving the problem and $r = o(\Pi)$.
The space usage of succinct data structures is measured only in terms of its
\emph{redundancy}. Using succinct data structures may be
crucial in applications where memory is scarce. We typically distinguish two types of
succinct data structures, namely \emph{systematic} and
\emph{non-systematic} data structures. Systematic data structures are more restricted
than non-systematic ones, in the sense that they always store the input
in read-only memory and then build an $r$-bit data structure on the
side. Non-systematic data structures just use at most $\Pi+r$ bits
(and thus do not have to store the input in the format in which it is
given). Succinct data structures have been studied extensively for
decades with many fundamental and important upper and lower bounds,
see e.g.~\cite{Jacobson:1988:SSD:915547, gal:succinct, patrascu08succinct, patrascu10ranklb}. The current strongest technique typically allows
one to prove lower bounds of the form $t r= \Omega(\Pi)$, see
e.g.~\cite{gal:succinct, BL13}.

The reason we take special interest in succinct data structures,
is that the matrix-vector multiplication data structure by Larsen and
Williams~\cite{LW17} is, in fact, a succinct data structure. In addition to
answering queries in just $O(n^{7/4}/\sqrt{w})$ time, it is
systematic and just stores the input matrix as read-only, plus an
additional $r= O(n^{7/4} \sqrt{w})$ bits on the side. With the current
techniques for proving lower bounds for succinct data structures, we
actually have hopes of proving something stronger than the $t =
\Omega(n)$ lower bounds we can hope for if we just consider general
data structures. Since $\Pi = n^2$ for Boolean matrix-vector
multiplication, it seems reasonable to hope for something
of the form $t r= \Omega(n^2)$. Such lower bounds would shed interesting new
light on this central data structure problem and would bring us closer
to understanding the true complexity of matrix-vector multiplication.

\subsection{Our Results}
Our main results are near-matching upper and lower bounds for
systematic succinct data structures solving Boolean
matrix-vector multiplication. On the upper bound side, we improve on
the results of Larsen and Williams and give a new randomized data
structure with the following guarantees:
\begin{theorem}
\label{thm:UB}
Given any matrix $M\in \{0,1\}^{n \times n}$ there exists a systematic
succinct data
structure $R(M)$ consisting of $r=O(n^{3/2} (\sqrt{w}+\frac{\lg
  n}{\sqrt{w}}))$ additional bits, and a query algorithm such that,
given any $v \in \{0,1\}^n$ it computes $Mv$ over the Boolean semiring
with probability $\ge 1 -1/n$ by probing at most $O(n^{3/2} \sqrt{w}\lg n)$ cells of $R(M)$ and $M$, where $w$ is the word size.
\end{theorem}
Our data structure thus reduces the redundancy
from $O(n^{7/4} \sqrt{w})$ bits to $O(n^{3/2} (\sqrt{w}+\frac{\lg
  n}{\sqrt{w}}))$ bits. 
Moreover, our randomized query algorithm returns the correct answer with high probability, and improves the query time over the previously known deterministic algorithm from $O(n^{7/4}/
\sqrt{w})$ to $O(n^{3/2}\sqrt{w} \lg n)$. 

We complement our new upper bound by a near-matching lower bound:
\begin{theorem}
\label{thm:mainMvBoolean}
Assume that for every matrix $M \in \{0,1\}^{n \times n}$ there exists
a systematic succinct data structure $R=R(M)$ consisting of at most
$r=r(n)$ bits and there is a randomized algorithm that given any $v
\in \{0,1\}^n$ computes $Mv$ over the Boolean semiring with probability $\ge 1-1/n$ by probing $R$ and at most $t=t(n)$ entries from $M$, 
Then for $n \le r \le n^2/4$, $t \cdot r  = \Omega(n^3)$; otherwise for $r  < n$, $t = \Omega(n^2)$.
\end{theorem}
Our lower bound comes within polylogarithmic factors of the upper
bound and is in fact higher than what we could hope for with previous
techniques (recall that previous techniques peak at $t r=
\Omega(\Pi)$). The proof of our lower bound exploits the large number $m$
of possible queries and we essentially manage to derive lower bounds of the form $t r=
\Omega(\Pi \lg m ) = \Omega(n^3)$. Also note that our lower bound
allows the data structure to probe all of $R$, i.e. all the redundant
bits, and still it says that one has to read a lot from the matrix
itself. Another exciting point is that our lower bound shows that $t =
\Omega(n^2)$ for any $r  < n$. Previous lower bounds of $t r =
\Omega(\Pi)$ always degenerate linearly in $r$ all the way down to
$r=1$. In contrast, our lower bounds say that one cannot do much better 
(up to a constant factor) than reading all $n^2$ entries of $M$ if $r < n$.

Finally, we also study matrix-vector multiplication over
$\mathbb{F}_2$. Here we prove lower bounds even for the
vector-matrix-vector multiplication where one is given a pair of
vectors $u,v \in \mathbb{F}_2^n$ as queries and must compute
$u^{\dagger}Mv$. This problem has just one bit in the output, making
it more difficult to prove lower bounds. Nonetheless, we prove the
following lower bound:
\begin{theorem}
\label{thm:mainUMVF2}
Assume that for every matrix $M \in \mathbb{F}_2^{n \times n}$ there exists a data structure $R=R(M)$ consisting of at most $r=r(n)$ bits and there is an algorithm that given $u,v \in \mathbb{F}_2^n$ computes $u^{\dagger}Mv$ by probing $R$ and at most $t=t(n)$ entries from $M$, then for $n \le r \le n^2/4$, $t \cdot r = \Omega(n^3/\lg n)$. Moreover, if $r < n$ then $t = \Omega(n^2/\lg n)$.
\end{theorem}
We believe it is quite remarkable that we can get $t = \tilde\Omega(n^2)$ lower
bounds for any $r < n$ for this 1-bit output problem. Since any $u^{\dagger}Mv$ query can be answered by just taking inner
product between $u$ and $Mv$, as a corollary of the above we also get
the same trade-off for matrix-vector problem over $\mathbb{F}_2$. To the best
of our knowledge, prior to this result there was no trade-off known
for such a small sized field. Our proof is completely information theoretic and based on an encoding argument. It is worth noting that our proof technique can be generalized to give lower bound for the case when the query algorithm may err with probability at most $1/64$ (though any small constant probability will work) on average over the random choices of $M,u,v$.

Finally, we also consider vector-matrix-vector multiplication over the
Boolean semiring. Since one can compute $Mv$ by running the following
sequence of $n$ queries:
$(e^{(1)})^{\dagger}Mv,\cdots,(e^{(n)})^{\dagger}Mv$ where
$\{e^{(i)}\}_{i \in [n]}$ is the standard basis over $\{0,1\}^n$, from
Theorem~\ref{thm:mainMvBoolean} we get a lower bound of $tr\ge \Omega(n^2)$ for the
Boolean $u^{\dagger}Mv$ problem. Instead of this simple reduction,
even if we use the much more elegant reduction in~\cite{HKNS15}, we will
not be able to derive any better lower bound from the above
theorem. However in Section~\ref{sec:Bool-uMv} we show how to extend
the proof of Theorem~\ref{thm:mainMvBoolean} to get a $t = \Omega(n/\lg n)$ bound on the
worst case number of probes into $M$ for the Boolean vector-matrix-vector
problem with $r \le n^2/4$.

\section{Preliminaries}
\label{sec:prelims}
\paragraph*{Notations. } For $k \in \mathbb{N}$, let $[k]$ denote the set $\{1,2,\ldots,k\}$. For every $v \in \{0,1\}^n$ and $i \in [n]$, let $v_i$ denote the $i$-th entry of $v$. All the logarithms we consider are over base $2$. We use the notation $x \in_R {\cal X}$ to denote that $x$ is drawn uniformly at random from the domain ${\cal X}$.

\paragraph*{Information Theory. }
Throughout this paper we use several basic definitions and notations from information theory. For further exposition readers may refer to any standard textbook on information theory (e.g.~\cite{CT06}).

Let $X,Y$ be discrete random variables on a common probability space. Let $p(x),p(y),p(x,y)$ denote $\Pr[X=x], \Pr[Y=y], \Pr[X=x,Y=y]$ respectively. 
The \emph{entropy} of $X$ is defined as
 $H(X):= - \sum_x p(x) \lg p(x)$. 
The \emph{joint entropy} of $(X,Y)$ is defined as $H(X,Y) \eqdef - \sum_{(x,y)} p(x,y) \lg (p(x,y))$. 
The \emph{mutual information} between $X$ and $Y$ is defined as $I(X;Y) \eqdef \sum_{(x,y)}{p(x,y) \lg \frac{p(x,y)}{p(x)p(y)}}$ and the \emph{conditional entropy} of $Y$ given $X$ is defined as $H(Y \mid X) \eqdef H(Y) - I(X;Y)$.

\begin{proposition}[Chain Rule of Entropy]
\label{prop:chnentro} 
Then $H(X,Y)=H(X)+H(Y \mid X)$.
\end{proposition}

The seminal work of Shannon~\cite{Sha48} establishes a connection between the entropy and the expected length of an optimal code encoding a random variable.
\begin{theorem}[Shannon's Source Coding Theorem~\cite{Sha48}]
Let $X$ be a discrete random variable over domain ${\cal X}$. Then for every uniquely decodable code $C :{\cal X} \to \{0,1\}^*$,  $\mathbb{E}(|C(X)|) \ge H(X)$. Moreover, there exists a uniquely decodable code $C:{\cal X} \to \{0,1\}^*$ such that $\mathbb{E}(|C(X)|) \le H(X)+1$.
\end{theorem}

\section{Upper Bound for Boolean Matrix-Vector Problem}
\label{sec:UB-Boolean-Mv}
In this section we prove Theorem~\ref{thm:UB} by introducing an efficient cell probe data structure for solving Boolean matrix-vector problem (with high probability). Let us first recall the theorem.

\begin{reptheorem}{thm:UB}
Given any matrix $M\in \{0,1\}^{n \times n}$ there exists a data structure $R$ consisting of $O(n^{3/2} (\sqrt{w}+\frac{\lg n}{\sqrt{w}}))$ bits, and a query algorithm such that, given $v \in \{0,1\}^n$ it computes $Mv$ with high probability by probing at most $O(n^{3/2} \sqrt{w}\lg n)$ cells of $R$ and $M$, where $w$ is the word size.
\end{reptheorem}

\paragraph{Preprocessing. }
In what follows, we present an algorithm that, given a matrix $M \in \{0,1\}^{n \times n}$ constructs the data structure guaranteed in Theorem~\ref{thm:UB}. Loosely speaking, the data structure is composed of a list ${\cal L}$ consisting of pairs $(I,J)$, and an encoding ${\cal E}$ of all the entries $(i,j) \in \bigcup_{(I,J) \in {\cal L}}(I \times J)$ such that $M_{(i,j)}=1$. 
The key step of the preprocessing algorithm is deciding which set-pairs to add to the list.
Informally, going over all possible pairs $(I,J)$, the algorithm adds a pair $(I,J)$ to ${\cal L}$, if there exists a large subset of $I \times J$ of entries not "covered" by the pairs already added in the list, and such that the "uncovered" part of the submatrix $M_{I,J}$ contains "few" $1$-entries per-row. The algorithm is formally described as Algorithm~\ref{alg:preprocess}.

\begin{algorithm}[H]
\begin{algorithmic}[1]
\STATE let ${\cal L} \leftarrow \emptyset$.
\FORALL{$(I,J) \in 2^{[n]}\times2^{[n]}$}
\STATE let ${\cal U} \eqdef \bigcup_{(I',J') \in {\cal L}}(I'\times J')$.
\STATE add $(I,J)$ to ${\cal L}$ if all of the following conditions hold. \label{l:cond}
\begin{enumerate}
	\item \label{itm:genSize} $\left| (I \times J) \setminus {\cal U} \right| \ge \frac{n^{3/2}}{\sqrt{w}}$; and
	\item \label{itm:density} $\forall i \in I$, $\Pr_{j \in J : (i,j) \notin {\cal U}} [M_{(i,j)}=1] \le \frac{1}{\sqrt{nw}}$.
\end{enumerate}
\ENDFOR
\STATE let ${\cal E}$ be the list of all $(i,j) \in \bigcup_{(I,J) \in {\cal L}}(I \times J)$ such that $M_{(i,j)}=1$.
\RETURN $({\cal L}, {\cal E})$
\caption{Preprocessing $M$}
\label{alg:preprocess}
\end{algorithmic}
\end{algorithm}

The following claim implies the first part of Theorem~\ref{thm:UB}.

\begin{claim}
\label{clm:sizeUB}
The string $({\cal L}, {\cal E})$ can be encoded using at most $O(n^{3/2} (\sqrt{w}+\frac{\lg n}{\sqrt{w}}))$ bits.
\end{claim}

\begin{proof}
Observe the conditions in line~\ref{l:cond} of the algorithm.
Due to condition~\ref{itm:genSize}, there are at most $n^{1/2}\sqrt{w}$ many different pairs in the list ${\cal L}$, and each pair can be encoded using only $2n$ bits (an indicator bit per row and column). Therefore ${\cal L}$ can be encoded using at most $O(n^{3/2} \sqrt{w})$ bits. Condition~\ref{itm:density} asserts that the density of $1$-entries in the submatrix covered by all the subsets of rows and columns listed in ${\cal L}$ is at most $\frac{1}{n^{1/2}\sqrt{w}}$. Each such entry can be encoded using $2\lg n$ bits and hence $\mathcal{E}$ can be encoded using at most $O(n^{3/2} \lg n/\sqrt{w})$ bits.
\end{proof}

\paragraph{Answering queries. }
To prove the second part of the theorem, we give a query algorithm that receives $v \in \{0,1\}^n$ and gets access to $M$, as well as to ${\cal L}$ and ${\cal E}$, and computes $Mv$ with high probability.
Let $J = \{j \in [n] : v_j=1\}$ be the set of columns of $M$ which are relevant for computing $Mv$, and let ${\cal U} \eqdef \bigcup_{(I',J') \in {\cal L}}(I' \times J')$ be the set of all matrix indices that appear in ${\cal L}$. 
Starting with the set $I = [n]$ of all possible rows, the algorithm "prunes" $I$ throughout the execution. Whenever an index $i$ is removed from $I$, the algorithm fixes $u_i \in \{0,1\}$. 
During the first step, the algorithm goes over ${\cal E}$. If for some $i \in I$, there exists $j \in J$ such that $(i,j)$ is encoded in ${\cal E}$, the algorithm sets $u_i = 1$ and removes $i$ from $I$.
During the second step, for every $i \in I$ the algorithm samples $2\sqrt{n} \lg n$ entries $(i,j)$ from the set $(\{i\} \times J) \setminus {\cal U}$. If $M_{(i,j)}=1$ for at least one of these entries, the algorithm sets $u_i = 1$ and removes $i$ from $I$.
During the third step, the algorithm examines the set ${\cal R} \eqdef \{(i,j) : i \in I\; , \; j \in J \; and \; (i,j)\notin {\cal U}\}$ of remaining entries. If this set has more than $O(n^{3/2}\lg n / \sqrt{w})$ elements, the algorithm reports "failure". Otherwise for every $i \in I$, the algorithm probes all entries $(i,j) \in {\cal R}$. If $M_{(i,j)}=1$ for at least one of these entries, the algorithm sets $u_i = 1$ and removes $i$ from $I$. Otherwise, the algorithm sets $u_i = 0$ and removes $i$ from $I$.
The algorithm terminates either by reporting "failure" or by returning $u = (u_1,\ldots,u_n)$. It is formally described as Algorithm~\ref{alg:query}.

\begin{algorithm}[th]
\begin{algorithmic}[1]
\STATE let $I \leftarrow [n], J \leftarrow \{j \in [n] : v_j =1\}$.
\STATE let ${\cal U} \leftarrow \bigcup_{(I',J') \in {\cal L}}(I'\times J')$.
\FORALL{$(i,j)$ encoded in ${\cal E}$} \label{l:firstStepStart}
\IF{$i \in I$ and $j \in J$}
\STATE set $u_i = 1$ and let $I \leftarrow I \setminus \{i\}$. \label{l:firstStepEnd}
\ENDIF
\ENDFOR
\FORALL{$i \in I$} \label{l:secondStepStart}
\STATE sample uniformly (with repetition) $2 \sqrt{nw} \lg n$ entries from $\{j \in J : (i,j) \notin {\cal U}\}$. \\ 
denote the set of sampled entries $J''(i)$.
\IF{there exists $j \in J''(i)$ such that $M_{(i,j)}=1$}
\STATE set $u_i \leftarrow 1$ and let $I \leftarrow I \setminus \{i\}$. \label{l:secondStepEnd}
\ENDIF
\ENDFOR
\STATE let ${\cal R} \leftarrow (I \times J) \setminus {\cal U}$. \label{l:thirdStepStart}
\IF{$|{\cal R}| \ge n^{3/2} / \sqrt{w}$} 
\RETURN "failure"
\ENDIF
\FORALL{$i \in I$}
\IF{there exists $j \in J$ such that $(i,j) \in {\cal R}$ and $M_{(i,j)} = 1$}
\STATE set $u_i \leftarrow 1$ and let $I \leftarrow I \setminus \{i\}$.
\ELSE
\STATE set $u_i \leftarrow 0$ and let $I \leftarrow I \setminus \{i\}$.
\ENDIF
\ENDFOR 
\RETURN $u = (u_1,\ldots,u_n)$
\caption{Querying $Mv$}
\label{alg:query}
\end{algorithmic}
\end{algorithm}

We will first show that the algorithm probes "few" bits.
Since $({\cal L},{\cal E})$ can be encoded using at most $O(n^{3/2} (\sqrt{w}+\frac{\lg n}{\sqrt{w}}))$ bits, and since the algorithm samples at most $2 \sqrt{nw} \lg n$ entries from each row of the matrix, we conclude the following.
\begin{lemma}
Algorithm~\ref{alg:query} probes at most $O(n^{3/2}\sqrt{w}\lg n)$ bits of $M,{\cal L}, {\cal E}$ throughout the execution.
\end{lemma}

To finish the proof of Theorem~\ref{thm:UB} we show that with high probability, the algorithm returns the correct answer. To this end, fix $v \in \{0,1\}^n$, and consider an execution of Algorithm~\ref{alg:query} on $v$.
Let $I_1, I_2$ be the set $I$ after the first step of the algorithm (lines~\ref{l:firstStepStart}-\ref{l:firstStepEnd}), and the second step of the algorithm (lines~\ref{l:secondStepStart}-\ref{l:secondStepEnd}) respectively. In these notations, ${\cal R} \leftarrow (I_2 \times J) \setminus {\cal U}$. Finally, let $$I_1^* \eqdef \left\{i \in I_1 : \Pr_{j \in J : (i,j) \notin {\cal U}} [M_{(i,j)}=1] > \frac{1}{\sqrt{nw}}\right\} \;.$$

\begin{lemma}
Algorithm~\ref{alg:query} fails with probability at most $\tfrac{1}{n}$. Moreover, if the algorithm does not fail, then it returns $Mv=u$.
\end{lemma}

\begin{proof}
Let ${\cal F}$ be the event $I_2 \subseteq I_1 \setminus I_1^*$. By definition of $I_1^*$ we get that 
\begin{equation*}
\begin{split}
\Pr[{\cal F}] \ge 1 - \sum_{i \in I_1^*}{\Pr[\forall j \in J''(i). \; M_{(i,j)}=0]} 
&\ge 1 - \sum_{i \in I_1^*}{\left(1 - \frac{1}{\sqrt{nw}}\right)^{2\sqrt{nw}\lg n}} \ge 1 - \frac{1}{n}
\end{split}
\end{equation*}
Conditioned on ${\cal F}$ occurring, for every $i \in I_2$, $\Pr_{j \in J : (i,j) \notin {\cal U}} [M_{(i,j)}=1] \le \frac{1}{\sqrt{nw}}$. Since $I_2 \times J \not\subseteq {\cal U}$, then $(I_2, J) \notin {\cal L}$. By the construction of ${\cal L}$ we therefore conclude that $|(I_2 \times J) \setminus {\cal U}| < \frac{n^{3/2}}{\sqrt{w}}$, and the algorithm does not fail.
We will show next that if the algorithm does not fail, then for every $i \in [n]$, $u_i = [Mv]_i$. First note that if $i \notin I_2$, then the algorithm finds $j \in J$ such that $M_{(i,j)}=1$, and therefore $u_i=1=[Mv]_i$. Otherwise, assume $i \in I_2$. Then for every $j \in J$, if $(i,j) \in {\cal U}$, then $M_{ij}=0$, since otherwise $(i,j)$ would be encoded in ${\cal E}$ and removed during the first step of the execution.
Therefore, for every $j \in J$, if $M_{(i,j)}=1$ then $(i,j) \in (I_2 \times J) \setminus {\cal U}$.
Since the algorithm does not fail, it goes over all entries in $(I_2 \times J) \setminus {\cal U}$, and therefore $[Mv]_i=1$ if and only if there exists $j \in J$ such that $M_{(i,j)}=1$, which in turn implies $u_i=1$. 
\end{proof}

\section{Matching Lower Bound on Boolean Matrix-Vector Problem}
\label{sec:Boolean-Mv}
In this section we consider the Boolean $Mv$ problem, where $Mv = (\vee_{j \in [n]}(M_{(i,j)} \wedge v_j))_{i \in [n]}$, and prove Theorem~\ref{thm:mainMvBoolean}. The presented bound matches the upper bound shown in the last section up to some (small) polylogarithmic factor. Although Theorem~\ref{thm:mainMvBoolean} allows the query algorithm to be randomized that returns right answer with high probability, for the sake of simplicity we first focus on the deterministic regime. In Section~\ref{sec:randomizedQueryLB} we refine the proof to hold for randomized query algorithms. 

\begin{theorem}
\label{thm:detMvBoolean}
Assume that for every matrix $M \in \{0,1\}^{n \times n}$ there exists a data structure $R=R(M)$ consisting of at most $r=r(n)$ bits and there is an algorithm that given any $v \in \{0,1\}^n$ computes $Mv$ by probing $R$ and at most $t=t(n)$ entries from $M$. Then for $n \le r \le n^2/4$, $t \cdot r  = \Omega(n^3)$; otherwise for $r  < n$, $t = \Omega(n^2)$.
\end{theorem}

To prove the theorem, we will define a family ${\cal M} \subseteq \{0,1\}^{n \times n}$ of matrices, and a set of queries $v^{(1)}, \ldots, v^{(4r/n)} \in \{0,1\}^n$ such that the following holds for every $M \in {\cal M}$. (1) The set of answers to the queries $Mv^{(1)}, \ldots, Mv^{(4r/n)}$ hold a large amount of information on $M$; and (2) one can succinctly "encode" the execution of the query algorithm on the respective sequence of queries. That is, there exists a short (in terms of $t,r$) string ${\cal E}$ such that given ${\cal E}, R$, one can emulate the query algorithm over the sequence of queries, and return the answers $Mv^{(1)}, \ldots, Mv^{(4r/n)}$.

For the rest of the section, we additionally assume $n \le r \le \sqrt{n^3/4}$. The proof for the case $\sqrt{n^3/4} \le r \le n^2/4$ is similar, and the differences will be discussed towards the end of the proof.

\paragraph{A Family of Input Matrices. }
Let ${\cal M} \subseteq \{0,1\}^{n \times n}$ be the family of all matrices $M \in \{0,1\}^{n \times n}$ with the following property: If each row of $M$ is divided into $r/n$ contiguous {\em blocks} each containing $n^2/r$ consecutive entries, then exactly one entry in each block is $1$, and the rest are all $0$s.
Consider a matrix $M \in {\cal M}$. Each of the $r$ blocks in $M$ contains exactly one $1$-entry out of $n^2/r$ entries. The following claim is thus implied from the definition of entropy. 
\begin{claim} \label{clm:info-hard}
Let $M \in_R {\cal M}$. Then $H(M) = r \lg \tfrac{n^2}{r}$.
\end{claim}

\paragraph{Encoding argument. }
Consider the following sequence of $4r/n$ vectors in $\{0,1\}^n$. For every $m \in [4r/n]$ define $v^{(m)} \in \{0,1\}^n$ such that $v^{(m)}_j = 1$ if and only if $$j (\modu \tfrac{n^2}{r}) \in \{\frac{(m-1)n^3}{4r^2}+1 , \ldots, \frac{mn^3}{4r^2}\} \;.$$

Fix some $M \in {\cal M}$. When querying for $Mv^{(1)}, \ldots, Mv^{(4r/n)}$, the algorithm reads at most $4tr/n$ bits from $M$.  
Observe that one can trivially encode all the probed entries using $\tfrac{4tr}{n} (2\lg n+1)$ bits by specifying the indices and the entry values. In turn, this encoding implies $t \ge \Omega(n/\lg n)$. However, by employing a subtler argument inspired by~\cite{BL13}, we provide a much better bound on $t$ and $r$, which also implies the simpler one. 

To this end, let $b,k$ denote the total number of different entries and the number of different $1$-entries read by the algorithm throughout this sequence of $4r/n$ many $Mv$ queries respectively. Then $k \le b \le 4tr/n$. 
Let $B$ be the sequence of $b$ different entries probed by the query algorithm, when queried for $Mv^{(1)}, \ldots, Mv^{(4r/n)}$, in the order they are probed. That is, $B : [b] \hookrightarrow [n] \times [n]$ is injective. Let $K \eqdef \{j \in [b] : M_{B(j)}=1\} \subseteq [b]$. Define ${\cal E}$ to be the bit-string composed of the three following sub-strings. The first $\lg(4tr/n)$ bits of ${\cal E}$ encode $b$. The next $\lg(4tr/n)$ bits encode $k$. The remaining bits encode $K$ as a subset of $[b]$. Since $|K|=k$, the following is straightforward. 

\begin{claim}
${\cal E}$ can be encoded using at most $\lg \binom{4tr/n}{k} + 2\lg (4tr/n)$ bits. 
\end{claim}
Next, we show that ${\cal E}$ and $R$ hold a "large amount" of information of $M$. 
\begin{lemma}
\label{lem:encode-one}
The bit-string ${\cal E}$ encodes the subset of $k$ $1$-entries among all the entries probed. Furthermore, given access to ${\cal E}$ and $R$ one can answer all the queries $Mv^{(1)}, \ldots, Mv^{(4r/n)}$.
\end{lemma}
\begin{proof}
First note, that given ${\cal E}$, one can directly decode $b,k$ and $K$.
We next show an emulation algorithm that, given access to $R, b, k ,K$ finds the $k$ $1$-entries in question, and moreover, answers all the queries $Mv^{(1)}, \ldots, Mv^{(4r/n)}$.
The algorithm emulates the query algorithm for $Mv^{(1)}, \ldots, Mv^{(4r/n)}$. Whenever the query algorithm probes a matrix entry, the emulation algorithm feeds it with an answer as follows.
If the entry has been previously probed during the execution, the query algorithm is fed with the same answer. Otherwise, the answer is determined to be $1$ or $0$ by whether the index of the matrix entry in the sequence of probes is in $K$ or not respectively, and the answer is stored for later probes. The algorithm is formally described as Algorithm~\ref{alg:emulation}.

\begin{algorithm}[H]
\begin{algorithmic}[1]
\STATE let $j \leftarrow 1$.
\STATE let $B^{alg} = \emptyset$.
\FOR{$i = 1,2,\ldots,4r/n$}
\STATE run the query algorithm for $Mv^{(i)}$
\STATE whenever the query algorithm probes a matrix entry $(p,q) \in [n]\times [n]$
\IF{there exists $j_0 < j$ such that $B^{alg}[j_0] = (p,q)$} \label{l:previouslyProbed}
\STATE $\ell \leftarrow j_0$
\ELSE
\STATE $B^{alg}[j] \leftarrow (p,q)$, $\ell \leftarrow j$ and $j \leftarrow j+1$.
\ENDIF
\IF{$\ell \in K$}
\STATE feed the query algorithm with $M_{(p,q)}=1$.
\ELSE
\STATE feed the query algorithm with $M_{(p,q)}=0$.
\ENDIF
\ENDFOR
\caption{Emulating a Sequence of Queries}
\label{alg:emulation}
\end{algorithmic}
\end{algorithm}

To prove the lemma, it is enough to show that the emulation algorithm always feeds the query algorithm with the correct answer. Denote the sequence of entry probes (with repetitions) performed by the query algorithm by $\{(p_m,q_m)\}_{m=1}^{4tr/n}$ (we may assume for simplicity that the query algorithm performs exactly $t$ entry probes for each query). The crux of the argument is that for all $j \in [b]$, $B^{alg}[j]=B(j)$. We prove the claim by induction on $m$. 
Clearly, during the first probe, $j=1$ and therefore the condition in line~\ref{l:previouslyProbed} of Algorithm~\ref{alg:emulation} is false. Therefore the algorithm sets $B^{alg}[1]$ to be $(p_1,q_1)$, which equals $B(1)$. Moreover, the emulation algorithm answers $1$ if and only if $1 \in K$, which occurs if and only if $M_{(p_1,q_1)}=1$. Assuming correctness for all $m'<m$, we will prove correctness for $m$. If there exists $j_0 < j$ such that $B^{alg}[j_0] = (p_m,q_m)$, then this entry probe has already been answered correctly by the induction hypothesis. Since the emulation algorithm gives the same answer as before, the answer is the correct one. Otherwise, this entry has not been probed yet, and therefore $B^{-1}((p_m,q_m)) = j$. The emulation algorithm then answers $1$ if and only if $j \in K$, which happens if and only if $M_{(p_m,q_m)}=1$. Therefore the emulation algorithm always gives the correct answer, thus finding the correct set of $k$ $1$-entries and answering all the queries $Mv^{(1)}, \ldots, Mv^{(4r/n)}$.
\end{proof}

The next lemma states that, in addition to learning the $k$ $1$-entries of $M$, by answering all the queries $Mv^{(1)}, \ldots, Mv^{(4r/n)}$, we can learn a lot of information about the other blocks of $M$.
\begin{lemma}
\label{lem:remain-information}
Let $M \in_R {\cal M}$. Then $H(M | R, {\cal E}) \le (r-k) \lg \tfrac{n^2}{4r}$.
\end{lemma}

\begin{proof}
By Lemma~\ref{lem:encode-one}, given access to only ${\cal E}$ and $R$, we can answer $Mv^{(1)}, \ldots, Mv^{(4r/n)}$, thus finding $k$ many $1$-entries of $M$. Moreover, for each block in $M$, the algorithm finds at least $3n^2/4r$ many $0$-entries.

By emulating the query algorithm for $Mv^{(1)}, \ldots, Mv^{(4r/n)}$, an algorithm can conclude from ${\cal E}$ and $R$ the exact locations of the $k$ many $1$-entries. Next, let $i \in [n]$ and let $m \in [4r/n]$. Note that if $[Mv^{(m)}]_i=0$, then for each block in the $i$-th row of $M$, all the entries corresponding to $1$-entries in $v^{(m)}$ must be $0$. Since there are exactly $n^3/4r^2$ such entries, and since  $[Mv^{(m)}]_i=1$ for at most $r/n$ values of $m$, the algorithm learns at least $\tfrac{3r}{n} \cdot \tfrac{n^3}{4r^2} = \tfrac{3n^2}{4r}$ $0$-entries in each block of $M$. Now the lemma follows from the Shannon's source coding theorem.
\end{proof}

We now turn to finish the proof of Theorem~\ref{thm:detMvBoolean}.
\begin{proof}[Proof of Theorem~\ref{thm:detMvBoolean}]
First observe that $r + \lg \binom{4tr/n}{k} +2\lg (4tr/n) \ge H(R,{\cal E}) \ge I(M;(R, {\cal E})) = H(M) -  H(M | R, {\cal E})$. By Lemma~\ref{lem:remain-information}, the right hand side is at least $ r \lg \tfrac{n^2}{r} - (r-k)\lg \tfrac{n^2}{4r} = 2r + k \lg \tfrac{n^2}{4r}$.
Rearranging we get that $\lg \binom{4tr/n}{k}+2\lg (4tr/n) - k \lg \tfrac{n^2}{4r} \ge r$. Since $\lg \binom{4tr/n}{k} \le k \lg \tfrac{4etr}{kn}$, we get that $k\lg \tfrac{16etr^2}{n^3k} +2\lg (4tr/n)\ge r$, and as for every $x, \alpha > 0$, $x \lg \tfrac{\alpha}{x} \le \alpha/2$ we have $\tfrac{8etr^2}{n^3}+2\lg (4tr/n) \ge r$, or $tr \ge \Omega(n^3)$.
This completes the proof for the case $n \le r \le  \sqrt{n^3/4}$. 

The proof for $\sqrt{n^3/4} \le r \le n^2/4$ is similar. The only difference is that in this case we choose $4r/n$ vectors in a slightly different way. For every $m \in [n^2/r]$ and $i \in [4r^2/n^3]$, we define a vector $v^{(m,i)}$ by setting $v^{(m,i)}_j=1$ if and only if $j (\modu \tfrac{n^2}{r}) = m$ and $(i-1)n^4/4r^2 +1 \le j \le in^4/4r^2$. The proof then follows in an analogous manner.

For the second part of the theorem, i.e. when $r < n$, we use a simple padding argument. If $r < n$, create a new data structure $R'$ by appending some arbitrary bits to $R$ such that $R'$ contains exactly $n$ bits. Now from the previous argument it follows that $t \ge \Omega (n^2)$.
\end{proof}

\subsection{Lower Bound for Randomized Query Algorithms}
\label{sec:randomizedQueryLB}
In this section we will extend the lower bound results shown previously to the case of randomized query algorithms, thus completing the proof of Theorem~\ref{thm:mainMvBoolean}. Assume that given a matrix $M\in \{0,1\}^{n \times n}$ there exists a data structure $R=R(M)$ consisting of at most $r=r(n)$ bits and there exists a {\em randomized} algorithm that, given $v \in \{0,1\}^n$ returns $Mv$ with probability $\ge 1 - \tfrac{1}{n}$ by probing $R$ and at most $t=t(n)$ entries from $M$.
\begin{reptheorem}{thm:mainMvBoolean}
If $n \le r \le n^2/4$, then $t\cdot r \ge \Omega(n^3)$; otherwise for $r < n$, $t \ge \Omega(n^2)$.
\end{reptheorem}
The proof will employ similar arguments to the proof of Theorem~\ref{thm:detMvBoolean}. We will therefore focus only on the case $n \le r \le \sqrt{n^3/4}$.  
Consider the vectors $v^{(1)}, \ldots, v^{(4r/n)}$ from the proof of Theorem~\ref{thm:detMvBoolean}. Fix some $M \in {\cal M}$, and run the query algorithm on $Mv^{(1)}, \ldots, Mv^{(4r/n)}$. For every $j \in [4r/n]$, denote the string of random bits used by the algorithm when queried for $Mv^{(j)}$ by $x_j$, and let $X \eqdef \left\langle x_1,\ldots,x_{4r/n} \right\rangle$. Let $F$ denote the indicator for the event that the query algorithm failed answering at least one of the queries. Applying union bound, $p \eqdef \Pr[F = 1] \le 1/3$. 
Let $B,b,K,k$ be as in the proof of Theorem~\ref{thm:detMvBoolean}. If $F=0$, we encode the string ${\cal E}$ the same way as in the previous proof and append $0$ as its first bit, whereas if $F=1$ we let ${\cal E}$ be an encoding of $M$ and append $1$ as its first bit.
In a similar manner to the one presented in Algorithm~\ref{alg:emulation}, given $R, {\cal E}$ and $X$, one can emulate the sequence of queries $Mv^{(1)}, \ldots, Mv^{(4r/n)}$ performed by the query algorithm, while using $x_1,\ldots,x_{4r/n}$ as random strings respectively. 
By computing $H(F,M \mid R,{\cal E},X)$ in two different ways we get the following.
\begin{equation*}
\begin{split}
H(F,M \mid R,{\cal E},X) &= H(M \mid R,{\cal E},X) + H(F \mid M, R,{\cal E},X) \\
&= H(F \mid R,{\cal E},X) + H(M \mid F, R,{\cal E},X)
\end{split}
\end{equation*}
Note that $H(F \mid R,{\cal E},X) \le H(F) < 1$ and $H(F \mid M, R,{\cal E},X) = 0$, since given ${\cal E}$ one learns also $F$.
Rearranging we get that 
$$H(M \mid R,{\cal E},X) \le 1 + H(M \mid F, R,{\cal E},X) \;.$$

\begin{equation}
\begin{split}
H(M \mid R,{\cal E},X) &\le 1 + H(M \mid F, R,{\cal E},X) \\
&= 1 + p \cdot  H(M \mid F=1, R,{\cal E},X) + (1-p) \cdot  H(M \mid F=0, R,{\cal E},X) \\
&\le 1 + 0 + (1-p)(r-k) \lg \tfrac{n^2}{4r} \\
\end{split}
\label{eq:remainingInfoRandom}
\end{equation}
where the last inequality follows from the same arguments as in Lemma~\ref{lem:remain-information}.

\begin{proof}[Proof of Theorem~\ref{thm:mainMvBoolean}]
First observe that from the Shannon's source coding theorem,
$$H({\cal E}) \le (1-p)\left(\lg \binom{4tr/n}{k} +2\lg \frac{4tr}{n}\right) + pr \lg \frac{n^2}{r}$$
and hence
\begin{equation*}
\begin{split}
r + (1-p)\left(\lg \binom{4tr/n}{k} +2\lg \frac{4tr}{n}\right) + pr \lg \frac{n^2}{r} &= H(R) + H({\cal E}) \\
&\ge H(R,{\cal E}) \\
&\ge I((M,X);(R,{\cal E})) = H(M,X) -  H(M,X| R, {\cal E}) \;.
\end{split}
\end{equation*}
Since $M,X$ are independent, we get that the r.h.s is at least $H(M) - H(M \mid R, {\cal E}, X)$. From \eqref{eq:remainingInfoRandom} we get that this is at least 
$$r \lg \tfrac{n^2}{r} - \left(1 + (1-p)(r-k)\lg \tfrac{n^2}{4r} \right) = pr \lg \tfrac{n^2}{r} + 2(1-p)r + (1-p) k \lg \tfrac{n^2}{4r} -1 \;.$$
Rearranging we get that 
$$\lg \binom{4tr/n}{k}+2\lg (4tr/n) - k \lg \tfrac{n^2}{4r} \ge 2r - r/(1-p) - 1/(1-p) \ge r/3 \;$$
By the same arguments as in the proof of Theorem~\ref{thm:detMvBoolean}, we get that $\tfrac{8etr^2}{n^3}+2\lg (4tr/n) \ge r/3$, or $tr \ge \Omega(n^3)$.
\end{proof}

\subsection{Lower Bounding Boolean Vector-Matrix-Vector Problem}
\label{sec:Bool-uMv}
In this section we extend the technique from Section~\ref{sec:Boolean-Mv} to get a lower bound on Boolean $u^{\dagger}Mv$ problem, albeit a weaker one. Assume that given a matrix $M \in \{0,1\}^{n \times n}$, there exists a data structure $R=R(M)$ containing at most $r=r(n)$ bits, and there exists an algorithm that, given $u,v \in \{0,1\}^n$ returns $u^{\dagger}Mv = (\vee_{i,j \in [n]}(M_{(i,j)} \wedge u_i \wedge v_j))$ while probing only $R$ and at most $t=t(n)$ bits from $M$.

 Let us first observe an easy corollary of Theorem~\ref{thm:mainMvBoolean}. Since answer of each $Mv$ query can be derived from the following sequence of $n$ queries: $(e^{(1)})^{\dagger}Mv,\cdots,(e^{(n)})^{\dagger}Mv$ where $\{e^{(i)}\}_{i \in [n]}$ is the standard basis over $\{0,1\}^n$, we get an lower bound of $tr\ge \Omega(n^2)$ for the Boolean $u^{\dagger}Mv$ problem for $n \le r \le n^2/4$. In this section we will prove a better lower bound for Boolean $u^{\dagger}Mv$ problem as stated in the following theorem.

\begin{theorem} \label{thm:lbBoolean-uMv}
If $r \le n^2/4$ then $t \ge \Omega(n/\lg n)$.
\end{theorem}
We prove the above theorem for $n \le r \le \sqrt{n^3/4}$. The proof extends for all $r < n^2/4$ by arguments similar to those used in the proof of Theorem~\ref{thm:detMvBoolean}. 
First note, that when comparing with the matrix-vector problem, the main caveat of the vector-matrix-vector problem is that each query results in exactly one bit, rather than $n$ bits.
Loosely speaking, we have observed that by querying $Mv^{(1)}, \ldots, Mv^{(4r/n)}$, where $v^{(1)}, \ldots, v^{(4r/n)}$ are defined as before, we gain a lot of information about $M$. This approach seems less beneficial when concerning vector-matrix-vector queries.
More specifically, it seems that we need $n$ times more queries to get the same amount of information.
By using the trivial argument demonstrated right before Theorem~\ref{thm:lbBoolean-uMv} we can not get our claimed bound. 
A more subtle observation into the proof of Lemma~\ref{lem:encode-one} shows that we can get a lot of information when the answer to the query is $0$. 
Particularly, assume $u^{\dagger}Mv = 0$. Then we know that whenever $u_i = v_j = 1$, $M_{(i,j)} = 0$. In what follows, we will need the following notation.

\begin{notation}
Let $x \in \{0,1\}^n$. Suppose $\overline{x} \in \{0,1\}^n$ denotes the {\em complement} of $x$. That is, $\overline{x}_j = 1$ if and only if $x_j = 0$ for every $j \in [n]$.
\end{notation}

Clearly, $\overline{x}^{\dagger}x = 0$, and moreover, $\overline{x}$ is the unique heaviest vector (in terms of Hamming weight) satisfying this property. 

\begin{lemma}
\label{lem:encode-uMv}
There exists a bit-string ${\cal E} = {\cal E}(M)$, 
such that ${\cal E}$ can be encoded using at most $\frac{4tr}{n} (2\lg n+1)$ bits, 
and moreover, given access to ${\cal E}$ and $R$ one can answer of $Mv^{(1)}, \ldots, Mv^{(4r/n)}$.
\end{lemma}
\begin{proof}
For every $j \in [4r/n]$, let $u^{(j)} \eqdef \overline{Mv^{(j)}}$. Take ${\cal E}$ to be the encoding of the list of all entries of $M$ probed throughout the sequence of queries $(u^{(1)})^{\dagger}Mv^{(1)}, \cdots , (u^{(4r/n)})^{\dagger}Mv^{(4r/n)}$, in the order they are probed. For simplicity, we may assume that the query algorithm probes exactly $t$ matrix entries for each query. For each entry we encode it's location in the matrix (using $2 \lg n$ bits), and its value (one more bit). Then ${\cal E}$ is composed of $4r/n$ "segments" of $t$ entries each. Clearly ${\cal E}$ can be encoded using at most $\frac{4tr}{n} (2\lg n+1)$ bits.
Let $j \in [4r/n]$. Now we provide an emulation algorithm that, given access to ${\cal E}$ and $R$, finds $u^{(j)}$, and thus finds $Mv^{(j)} = \overline{u^{(j)}}$. The emulation algorithm starts by setting $u^{(*)}$ to be all $0$ vector and goes over all $u \in \{0,1\}^n$ and emulates the query algorithm for $u^{\dagger}Mv^{(j)}$. Whenever the query algorithm probes an entry of $M$, the emulation algorithm looks for the encoding of this entry in ${\cal E}$. If it finds this entry, it feeds it to the query algorithm. Otherwise, it breaks and continues to the next $u \in \{0,1\}^n$. If the query algorithm terminates, and the answer is $0$, then the algorithm compares the Hamming weight of $u$ (denoted as $wt(u)$) with that of $u^{(*)}$. If  $wt(u) > wt(u^{(*)})$ the algorithm replaces $u^{(*)}$ with $u$.
The algorithm is formally given as Algorithm~\ref{alg:finduj}.

\begin{algorithm}[t]
\begin{algorithmic}[1]
\STATE let $u^{(*)} \leftarrow 0^n$.
\FORALL{$u \in \{0,1\}^n$}
\STATE run the query algorithm for $u^{\dagger}Mv^{(j)}$
\STATE whenever the query algorithm probes a matrix entry $(p,q) \in [n]\times [n]$
\IF{$(p,q)$ is encoded in ${\cal E}$}
\STATE feed the query algorithm with the corresponding bit in ${\cal E}$ as $M_{(p,q)}$.
\ELSE
\STATE stop the query algorithm.
\ENDIF
\IF{the query algorithm outputs $u^{\dagger}Mv^{(j)}=0$ and $wt(u) > wt(u^{(*)})$}
\STATE $u^{(*)} \leftarrow u$.
\ENDIF
\ENDFOR
\RETURN $u^{(*)}$
\caption{Finding $u^{(j)}$}
\label{alg:finduj}
\end{algorithmic}
\end{algorithm}

It is straightforward that $u^{(*)}$ is the unique $u \in \{0,1\}^n$ that satisfies the following.
\begin{enumerate}
\item  Every entry probed by the query algorithm when querying $u^{\dagger}Mv^{(j)}$ is encoded in ${\cal E}$ (thus given access to ${\cal E}, R$, one can answer $u^{\dagger}Mv^{(j)}$);
\item $u^{\dagger}Mv^{(j)}=0$; and
\item $u$ is of maximal Hamming weight.
\end{enumerate}
Therefore, $u^{(*)} = u^{(j)}$.
\end{proof}

We can now prove Theorem~\ref{thm:lbBoolean-uMv} using similar arguments to those in the proof of Theorem~\ref{thm:detMvBoolean}. 
Suppose while probing entries of the matrix $M$ we read total $k$ $1$-entries. Then by following the argument of the proof of Theorem~\ref{thm:detMvBoolean}, we get that
\begin{align*}
r + \tfrac{4tr}{n} (2\lg n+1) & \ge H(R,{\cal E}) \ge H(M) -  H(M | R, {\cal E})\\
&\ge 2r + k \lg \tfrac{n^2}{4r} \ge 2r.
\end{align*}
Now by rearranging the terms, we get that $t \ge \Omega(n/\lg n)$. $\hfill \Box$

Using the argument similar to that in Section~\ref{sec:randomizedQueryLB} we can also extend Theorem~\ref{thm:lbBoolean-uMv} to randomized query algorithms.
\begin{theorem}
\label{thm:randomizedLB-uMv}
Assume that for every matrix $M\in \{0,1\}^{n \times n}$ there exists a data structure $R$ consisting of at most $r=r(n)$ bits and a query algorithm that can answer any $u^{\dagger}Mv$ query with error probability at most $1/n$ by probing $R$ and at most $t=t(n)$ entries of $M$. Then for $r \le n^2/4$, $t \ge \Omega(n/\lg n)$.
\end{theorem}

\paragraph*{Note:} Though in both Theorem~\ref{thm:mainMvBoolean} and Theorem~\ref{thm:randomizedLB-uMv} we consider the error probability to be at most $1/n$, one can easily generalize the results for error probability to be any $1/n < \epsilon <1$. However we will lose extra $\lg n$ factor in the lower bound. More specifically, given any randomized algorithm $\mathcal{A}$ with probability of error $\epsilon$, we can boost the success probability to $(1-1/n)$ by repeating $\mathcal{A}$ $O(\lg n/ \lg (\frac{1}{\epsilon}))$ times and taking the majority vote. Now let us denote the new algorithm to be $\mathcal{B}$. Observe that $\mathcal{B}$ probes at most $O(t \lg n/ \lg (\frac{1}{\epsilon}))$ cells of the matrix and thus we will lose $O(\lg n/ \lg (\frac{1}{\epsilon}))$ factor in all the bounds given in Theorem~\ref{thm:mainMvBoolean} and Theorem~\ref{thm:randomizedLB-uMv}.

\ifpdf
\section{Lower Bound on Vector-Matrix-Vector Problem over \texorpdfstring{$\mathbb{F}_2$}{F2}}
\else
\section{Lower Bound on Vector-Matrix-Vector Problem over $\mathbb{F}_2$}
\fi
\label{sec:F2-uMv}
This section is devoted to the proof of Theorem~\ref{thm:mainUMVF2}. To this end, assume that given a matrix $M \in \mathbb{F}_2^{n \times n}$, there exists a data structure $R=R(M)$ consisting of at most $r=r(n)$ bits, and there exists an algorithm that, given $u,v \in \mathbb{F}_2^n$ returns $u^{\dagger}Mv$ while probing only $R$ and at most $t=t(n)$ bits from $M$. Under these assumptions, Theorem~\ref{thm:mainUMVF2} states the following.

\begin{reptheorem}{thm:mainUMVF2}
If $n \le r \le \tfrac{n^2}{64}$ then $t \cdot r = \Omega(n^3/\lg n)$; otherwise for $r < n$, $t = \Omega(n^2/ \lg n)$.
\end{reptheorem}

To prove the theorem we will show that for most matrices $M \in \mathbb{F}_2^n$, one can succinctly (in terms of $r,t$) encode $M$. 
More precisely, by fixing some parameter $B$ and dividing $M$ into segments of $B$ consecutive rows, we will show that there is a single segment that contains a large amount of information, and moreover, there exists a short (in terms of $r,t$) bit-string that encodes this segment.
It is worth noting that our proof technique can be generalized to give lower bound for the case when the query algorithm may err with probability at most $1/64$ on average over the choices of $M,u,v$. However for the sake of simplicity we first focus only on the query algorithm that never errs and we defer the comment on the generalization to the end of this section.

We start with introducing some notations. Given $u \in \mathbb{F}_2^n$ and a subset $I \subseteq [n]$, let $u_I$ be the projection of $u$ onto $\mathbb{F}_2^{|I|}$. Similarly denote $M_{I,J}$ for any $M \in \mathbb{F}_2^{n \times n}$ and subsets $I,J \subseteq [n]$.

Let $B$ be some parameter, the value of which will be fixed later. 
For every $i \in [n/B]$, let $I_i \eqdef \{(i-1)B+1, \ldots, iB\}$, and let $M_i \eqdef M_{I_i,[n]}$ be the $i$-th segment of $M$ composed of all the rows in $I_i$, and $M_{-i} \eqdef M_{([n] \setminus I_i), [n]}$.
Finally, for every $u,v \in \mathbb{F}_2^n$, $M \in \mathbb{F}_2^{n \times n}$ and  $i \in [n/B]$, let $t_i(u,M,v)$ be the number of cell probes performed by the algorithm in $M_i$ when queried for $u^{\dagger}Mv$.

Our first claim shows that there exists an $\ind \in [n/B]$ such that for many vectors $u \in \mathbb{F}_2^n$, the expected number (over random $M,v$) of probes performed by the algorithm on $M_{\ind}$ is not too large.
\begin{lemma}
\label{lem:expected-probe}
There exists $\ind \in [n/B]$ such that $\Pr_u[\mathbb{E}_{M,v}[t_{\ind}(u,M,v)] \le \frac{4tB}{n}] \ge \frac{3}{4}$.
\end{lemma}

\begin{proof}
First note that $\mathbb{E}_i[\mathbb{E}_{u,M,v}[t_i(u,M,v)]] = \mathbb{E}_{u,M,v}[\mathbb{E}_i[t_i(u,M,v)]] \le t$. Therefore there exists $\ind \in [n/B]$ such that $\mathbb{E}_{u,M,v}[t_{\ind}(u,M,v)] \le \frac{tB}{n}$.
The claim now follows from Markov's inequality.
\end{proof}

For every $u \in \mathbb{F}_2^n$ and $i \in [n/B]$, denote $M_{i|u} \eqdef u_{I_i}^{\dagger}  M_i$.
The following lemma shows that for most vectors $u \in \mathbb{F}_2^n$, $M_{\ind | u}$ contains a large amount of information. One may note that the lemma is true for all $i \in [n/B]$, though for our purpose it suffices to consider $\ind$ only.
\begin{lemma}
\label{lem:conditional-info}
Suppose $M \in_R \mathbb{F}_2^{n \times n}$. Then $\Pr_u[H(M_{\ind | u} \mid R, M_{- \ind}) \ge n - \frac{8r}{B}] \ge \frac{3}{4}$.
\end{lemma}
\begin{proof}
Let $U = \{u \in \mathbb{F}_2^n : H(M_{\ind | u} \mid R, M_{- \ind}) < n - \frac{8r}{B} \}$, and let $u^{(1)},\ldots,u^{(\ell)} \in U$ be a sequence in $U$ such that $u^{(1)}_{I_{\ind}}, \ldots, u^{(\ell)}_{I_{\ind}}$ are linearly independent over $\mathbb{F}_2^{|I_{\ind}|}$. Then the random variables $M_{\ind | u^{(1)}}, \ldots, M_{\ind | u^{(\ell)}}, M_{- \ind}$ are independent.
To see this, first note that by the definition ,$M_{-\ind}$ is independent of $M_{\ind}$ and hence of $M_{\ind | u^{(1)}}, \ldots, M_{\ind | u^{(\ell)}}$.
Next observe that for any $k \in [\ell]$ and $b \in \mathbb{F}_2^{n}$, $Pr_M[M_{\ind|u^{(k)}}=b]=1/2^n$. Now since the vectors $u^{(1)},\ldots,u^{(\ell)}$ are linearly independent, for any $b^{(1)},\ldots,b^{(\ell)} \in \mathbb{F}_2^{n}$, $Pr_M[\text{for all }k,\;M_{\ind|u^{(k)}}=b^{(k)}]=(1/2^n)^{\ell}$. 
Therefore
\begin{equation*}
\begin{split}
r \ge H(R) &\ge I(R ; M_{\ind | u^{(1)}}, \ldots, M_{\ind | u^{(\ell)}} \mid M_{- \ind}) \\
& = H(M_{\ind | u^{(1)}}, \ldots, M_{\ind | u^{(\ell)}} \mid M_{- \ind}) - H(M_{\ind | u^{(1)}}, \ldots, M_{\ind | u^{(\ell)}} \mid R, M_{- \ind}) \\
& \ge \sum_{j=1}^\ell{\left( H(M_{\ind | u^{(j)}}\mid M_{- \ind}) - H(M_{\ind | u^{(j)}} \mid R, M_{- \ind}) \right)} \\
& \ge \sum_{j=1}^\ell{\left( n - \left( n - \frac{8r}{B} \right) \right)} = \frac{8r \ell}{B} \;,
\end{split}
\end{equation*}
thus $\ell \le B/8$ implying $| \{ u_{I_{\ind}} : u \in U \} | \le 2^{B/8}$ and we get that $|U| \le 2^{n-2}$.
\end{proof}

Now the following is a simple application of union bound.
\begin{corollary} \label{cor:uExsists}
There exists $u^* \in \mathbb{F}^n_2$ such that
$$H(M_{\ind | u} \mid R, M_{- \ind}) \ge n - 8r/B \quad and \quad \mathbb{E}_{M,v}[t_{\ind}(u,M,v)] \le 4tB/n \;.$$
\end{corollary}
Let us define ${\cal M} \eqdef \{M \in \mathbb{F}_2^{n \times n} : \mathbb{E}_v[t_{\ind}(u^*,M,v)] \le \frac{8tB}{n}\}$ and then Markov's inequality implies the following.
\begin{claim}
\label{clm:F2InitialInfo}
$\Pr_M[M \in {\cal M}] \ge \frac{1}{2}$.
\end{claim}
The next Lemma shows that whenever $M \in {\cal M}$, $M_{\ind \mid u^*}$ can be encoded using a few bits.
\begin{lemma}
\label{lem:encode}
If $M \in {\cal M}$ then $M_{\ind \mid u^*}$ can be encoded using $\frac{64tB}{n} \lg n$ extra bits (in addition to $R$ and $M_{- \ind}$).
\end{lemma}
\begin{proof}
Fix some $M \in {\cal M}$. Then by Markov's inequality, $\Pr_v [t_{\ind}(u^*,M,v) \le \frac{16tB}{n}] \ge \frac{1}{2}$.
Therefore there exists a set ${\cal S}$ of $\frac{16tB}{n}$ entries in $M_{\ind}$ (note that the submatrix $M_{\ind}$ is of size $nB$) such that by probing only entries from $M_{- \ind}, R$ and ${\cal S}$ the algorithm can answer at least 
$$\frac{2^{n-1}}{\binom{nB}{16tB/n}} \ge \frac{2^{n-1}}{\left(\frac{enB}{16tB/n}\right)^{\frac{16tB}{n}}} = 2^{n-1 - \frac{16tB}{n} \lg \frac{en^2}{16t}} \ge 2^{n - \frac{32tB}{n} \lg n}$$
queries of the form $(u^*)^{\dagger}Mv$. The set ${\cal S}$ can be encoded using $\frac{16tB}{n} \lg (n^2)$ bits.
Let $V \subseteq \mathbb{F}_2^n$ denote the set of vectors such that for any $v \in V$ the query $(u^*)^{\dagger}Mv$ can be answered by probing entries only from $M_{- \ind}, R, {\cal S}$. Observe that $V$ is a linear subspace of $\mathbb{F}_2^n$, and $\dim(V) \ge n - \frac{32tB}{n} \lg n$.

Next, fix an ordering $v^{(1)},\ldots,v^{(2^n)}$ of $\mathbb{F}_2^n$, and consider the string ${\cal E}$ of bits constructed as follows. Starting with an empty string ${\cal E}$, for every $k \in [2^n]$, if $v^{(k)} \notin  {\sf span}\left( V \cup \{v^{(1)},\ldots,v^{(k-1)}\} \right)$, append $(u^*)^{\dagger}Mv^{(k)}$ to ${\cal E}$. Since $\dim(V) \ge n - \frac{32tB}{n} \lg n$, it follows that ${\cal E}$ can be encoded using at most $\frac{32tB}{n} \lg n$ bits.

Now by probing only $M_{- \ind}, R, {\cal S}, {\cal E}$ we can answer $(u^*)^{\dagger}Mv$ for all $v \in \mathbb{F}_2^n$, which in terms suffices to retrieve the string $M_{\ind \mid u^*}$. 
\end{proof}
 
 Now we are ready to prove the main result of this section. 
 \begin{proof}[Proof of Theorem~\ref{thm:mainUMVF2}]
We prove the theorem by showing that $n - \frac{8r}{B} \le \frac{64tB}{n}\lg n + \frac{n}{2} + 1$. Setting $B = \lfloor \frac{32r}{n}\rfloor$ then implies the theorem.

To this end, let $\mathbbm{1}_{M \in {\cal M}}$ denote the indicator random variable for the event $M \in {\cal M}$. 
By definition of $u^*$ we have 
\begin{equation}
n - \frac{8r}{B} \le H(M_{\ind | u^*} \mid R, M_{- \ind}) \le H(M_{\ind | u^*}, \mathbbm{1}_{M \in {\cal M}} \mid R, M_{- \ind}) \;.
\label{eq:addInd}
\end{equation}
Applying the chain rule of entropy we get that 
\begin{equation}
\begin{split}
H(\mathbbm{1}_{M \in {\cal M}}, M_{\ind | u^*} \mid R, M_{- \ind}) &= H(\mathbbm{1}_{M \in {\cal M}} \mid R, M_{- \ind}) + H(M_{\ind | u^*} \mid R, M_{- \ind}, \mathbbm{1}_{M \in {\cal M}}) \\
&\le H(\mathbbm{1}_{M \in {\cal M}}) + H(M_{\ind | u^*} \mid R, M_{- \ind}, \mathbbm{1}_{M \in {\cal M}}) \;.
\label{eq:chainRule}
\end{split}
\end{equation}
Clearly, $H(\mathbbm{1}_{M \in {\cal M}}) \le 1$. Next we bound $H(M_{\ind | u^*} \mid R, M_{- \ind}, \mathbbm{1}_{M \in {\cal M}})$ as follows.
\begin{equation}
\begin{split}
&H(M_{\ind | u^*} \mid R, M_{- \ind}, \mathbbm{1}_{M \in {\cal M}}) = \\
&=H(M_{\ind | u^*} \mid R, M_{- \ind}, \mathbbm{1}_{M \in {\cal M}}=1) \cdot \Pr[\mathbbm{1}_{M \in {\cal M}}=1] + H(M_{\ind | u^*} \mid R, M_{- \ind}, \mathbbm{1}_{M \in {\cal M}}=0) \cdot \Pr[\mathbbm{1}_{M \in {\cal M}}=0].
\end{split}
\label{eq:totalEntropy}
\end{equation}
Conditioned on $\mathbbm{1}_{M \in {\cal M}}=1$, Lemma~\ref{lem:encode} guarantees that we can encode $M_{\ind | u^*}$ using at most $\frac{64tB}{n}\lg n$ bits in addition to $R, M_{- \ind}$. Therefore by the Shannon's source coding theorem 
$$H(M_{\ind | u^*} \mid R, M_{- \ind}, \mathbbm{1}_{M \in {\cal M}}=1) \cdot \Pr[\mathbbm{1}_{M \in {\cal M}}=1] \le \frac{64tB}{n}\lg n \cdot 1 \;.$$
Claim~\ref{clm:F2InitialInfo} implies that 
$$H(M_{\ind | u^*} \mid R, M_{- \ind}, \mathbbm{1}_{M \in {\cal M}}=0) \cdot \Pr[\mathbbm{1}_{M \in {\cal M}}=0] \le n \cdot \frac{1}{2} \;.$$
Plugging the last two inequalities into \eqref{eq:totalEntropy} we get that 
$H(M_{\ind | u^*} \mid R, M_{- \ind}, \mathbbm{1}_{M \in {\cal M}}) \le \frac{64tB}{n}\lg n + \frac{n}{2} \;.$
Plugging this into \eqref{eq:addInd}, \eqref{eq:chainRule} we get that $n - \frac{8r}{B} \le \frac{64tB}{n}\lg n + \frac{n}{2} + 1$.
Now for $r \ge n$ by substituting $B=\lfloor \frac{32r}{n}\rfloor$, we conclude that $\frac{2048tr}{n^2}\lg n \ge \frac{n}{4}$, and thus $tr \ge \Omega (n^3/\lg n)$.

For $r <n$ we use the following simple padding argument. Append $R$ with some arbitrary bits so that the size (no. of bits) of the new data structure $R'$ becomes $n$. Now from the previous argument it follows that $t \ge \Omega (n^2/\lg n)$, thus completing the proof.
\end{proof}

\paragraph*{Comment on query algorithms with error. }
In Theorem~\ref{thm:mainUMVF2} we consider query algorithms those always output the correct answer and provide lower bound. It is worth noting that our proof technique can be generalized to give lower bound for the case when the query algorithm may err with probability at most $1/64$ (though any small constant probability will work) on average over the choices of $M,u,v$. We need to modify the proof a bit by considering the event that the algorithm (say ${\cal A}$) errs, i.e., ${\cal A}(u,M,v) \ne u^{\dagger}Mv$. From $Pr_{u,M,v}[{\cal A}(u,M,v) \ne u^{\dagger}Mv]\le 1/64$, using Markov's inequality we can deduce that $Pr_u[Pr_{M,v}[{\cal A}(u,M,v) \ne u^{\dagger}Mv] \ge 1/16] \le 1/4$. Now we choose $u^*$ that satisfies Corollary~\ref{cor:uExsists} and $Pr_{M,v}[{\cal A}(u^*,M,v) \ne (u^*)^{\dagger}Mv] \le 1/16$. The existence of such a $u^*$ follows from simple union bound. Similarly $Pr_M[Pr_{v}[{\cal A}(u^*,M,v) \ne (u^*)^{\dagger}Mv] \ge 1/4] \le 1/4$. Now define ${\cal M} \eqdef \{M \in \mathbb{F}_2^{n \times n} : \mathbb{E}_v[t_{\ind}(u^*,M,v)] \le \frac{16tB}{n}\} $ and hence $Pr_M[M \in {\cal M} \text{ and }Pr_{v}[{\cal A}(u^*,M,v) \ne (u^*)^{\dagger}Mv] \le 1/4] \ge 1/2$. Next we modify Lemma~\ref{lem:encode} by saying that for any $M \in {\cal M}$,
$$\Pr_v [{\cal A}(u^*,M,v) = (u^*)^{\dagger}Mv\text{ and }t_{\ind}(u^*,M,v) \le \frac{64tB}{n}] \ge \frac{1}{2}.$$ The remaining argument will be the same and we will get similar lower bound. One can further extend this lower bound result to randomized query algorithms that given $u,v$ output correct answer with high probability, by using the technique described in Section~\ref{sec:randomizedQueryLB}.

\bibliographystyle{alphaurlinit}
\bibliography{LB}
\end{document}